\documentclass[12pt]{article} \usepackage{fullpage}
\usepackage{amssymb}
\usepackage{amsmath}
\usepackage{amsthm}

\usepackage{algorithm}
\usepackage{algpseudocode}
\usepackage{nicefrac}
\usepackage{url}
\usepackage{color}
\usepackage[american]{babel}
\usepackage{graphicx}
\usepackage{subfigure}

\newtheorem{conjecture}{Conjecture}
\newtheorem{theorem}{Theorem}
\newtheorem{corollary}{Corollary}
\newtheorem{definition}{Definition}

\newtheorem{observation}{Observation}

\newtheorem{lemma}{Lemma}

\newcommand{\idx}{\text{idx}}

\newcommand{\leaveout}[1]{}

\renewcommand{\medskip}{\smallskip}
\renewcommand{\int}{{\ensuremath{\rm int\,}}}

\usepackage{lineno}

\date{\today}
\title{Trees and co-trees with constant maximum degree \\
in planar 3-connected graphs}
\author{Therese Biedl\thanks{David R.~Cheriton School of Computer
Science, University of Waterloo, Waterloo, Ontario N2L 1A2, Canada.
Supported by NSERC and the Ross and Muriel Cheriton Fellowship.
Research initiated while participating at Dagstuhl seminar 13421.}
}

\begin{document}

\maketitle
\begin{abstract}
This paper considers the
conjecture by Gr{\"u}nbaum that every planar 3-connected graph
has a spanning tree $T$ such that both $T$ and its co-tree have 
maximum degree at most 3.  Here, the {\em co-tree} of $T$ is
the spanning tree of the dual obtained by taking the duals of
the non-tree edges.  While Gr\"unbaum's conjecture remains
open, we show that every planar 3-connected graph has a spanning
tree $T$ such that both $T$ and its co-tree have
maximum degree at most 5.  It can be found in linear time.
\end{abstract}

\section{Introduction}

In 1966, Barnette showed that
every planar 3-connected graph has a spanning tree with maximum degree
at most 3  \cite{Bar66}.  (In the following, a {\em $k$-tree} denotes
a tree with maximum degree at most $k$.)   Since the dual of a
3-connected planar graph is also 3-connected, the dual graph $G^*$ also
has a spanning 3-tree.  Gr\"unbaum \cite{Gru70} conjectured that there are
spanning 3-trees in the graph and its dual that are simultaneous in the
sense of being tree and co-tree.  For any spanning tree $T$
in a planar graph, define the  {\em co-tree} to be the subgraph
of the dual graph formed by taking the dual edges of
the edges in $G-T$.  Since cuts in planar graphs correspond
to union of cycles in the dual graph, it is easy to see that the co-tree is
a spanning tree of $G^*$.  In 1970, Gr\"{u}nbaum posed the following conjecture:

\begin{conjecture}\cite{Gru70}
Every planar 3-connected graph has a spanning 3-tree for which
the co-tree is a spanning 3-tree of the dual graph.
\end{conjecture}

This conjecture was still open in 2007 \cite{Gru07}, and to our knowledge
remains open today.  This paper proves a slightly weaker statement:
Every planar 3-connected graph has a spanning 5-tree for which
the co-tree is a spanning 5-tree of the dual graph.

Our approach is to read this spanning 5-tree from the {\em canonical ordering},
a decomposition that exists for all 3-connected planar graphs \cite{Kant96}
and that has properties useful for many algorithms for graph drawing (see
e.g. \cite{CK97,Kant96,NR04}) and other applications (see e.g. \cite{HKL99}).
This will be formally defined in Section~\ref{se:definition}.  There are
readily available implementations for finding a canonical ordering 
(see for example \cite{boost,pigale}), 
and getting our tree from the canonical ordering is nearly trivial,
so our trees not only can be found in linear time, but it would be
very easy to implement the algorithm.

The canonical ordering is useful for Barnette's theorem as well.
Barnette's proof \cite{Bar66}
is constructive, but the algorithm that can be derived from the proof
likely has quadratic run-time (he did not analyze it).  With
a slightly more structured proof and suitable data structures, it is
possible to find the 3-tree in linear time \cite{Strothmann-PhD}.  But
in fact, the 3-tree can be directly read from the canonical ordering.
This was mentioned by Chrobak and Kant in their
technical report \cite{CK-TR93}, but no details were given as to why
the degree-bound holds, and they did not include the result in
their journal version \cite{CK97}.
We provide these details in Section~\ref{se:Barnette},
somewhat as a warm-up and because the key lemma will be needed later.
Then we prove the weakened version of
Gr\"unbaum's conjecture in Section~\ref{se:Gruenbaum}.

\section{Background}
\label{se:definition}

Assume that $G=(V,E)$ is a planar graph, i.e., it can be drawn
in the plane without crossing.  Also assume that $G$ is
3-connected, i.e., for any two vertex $\{u,v\}$ the graph resulting
from deleting $u$ and $v$ is still connected.  By Whitney's theorem
a 3-connected planar graph $G$ has a unique {\em combinatorial embedding},
i.e., in any planar drawing of $G$ the circular clockwise order of edges
around each vertex $v$ is the same, up to reversal of all these orders.
Given a planar drawing $\Gamma$, a {\em face} is a maximal connected
region of $\Bbb{R}^2-\Gamma$.  The unbounded face is called the
{\em outer-face}, all other faces are {\em interior faces}.  

Define the {\em dual graph} $G^*$ 
as follows.  For every face $f$ in $G$, add a vertex $f^*$
to $G^*$.  If $e$ is an edge of $G$ with incident faces $f_\ell$ and
$f_r$, then add edge $e^*:=(f_\ell^*,f_r^*)$ to $G^*$; $e^*$ is called
the {\em dual edge} of $e$.  

De Fraysseix, Pach and Pollack \cite{FPP90} were the first to introduce
a canonical ordering for triangulated planar graphs.
Kant \cite{Kant96} generalized the canonical ordering to all 3-connected
planar graphs.  

\begin{definition}\cite{Kant96}
A {\em canonical ordering} of a planar graph $G$ with a fixed combinatorial
embedding and outer-face is an ordered
partition $V=V_1\cup \dots \cup V_K$ that satisfies the following:
\begin{itemize}
\item $V_1$ consists of two vertices $v_1$ and $v_2$ where $v_2$
	is the counter-clockwise neighbour of $v_1$ on the outer-face.
\item $V_K$ is a singleton $\{v_n\}$ where $v_n$ is the
	clockwise neighbour of $v_1$ on the outer-face.
\item For each $k$ in $2,\dots,K$, the graph $G[V_1\cup \dots \cup V_k]$
	induced by $V_1\cup \dots \cup V_k$ is 2-connected and contains
	edge $(v_1,v_2)$ and all vertices of $V_k$ on the outer-face.
\item For each $k$ in $2,\dots,K-1$ one of the two following
	conditions hold:
	\begin{enumerate}
	\item $V_k$ contains a single vertex $z$ that has at least
		two neighbours in $V_1\cup \dots \cup V_{k-1}$
		and at least one neighbour in $V_{k+1}\cup \dots \cup V_K$.
	\item $V_k$ contains $\ell\geq 2$ vertices that induce
		a path $z_1-z_2-\dots-z_\ell$, enumerated in clockwise
		order around the outer-face of $G[V_1\cup \dots \cup V_{k}]$.
		Vertices $z_1$ and $z_\ell$
		have exactly one neighbour each in $V_1\cup \dots \cup V_{k-1}$,
		while $z_2,\dots,z_{\ell-1}$ have no such neighbours.
		Each $z_i$, $1\leq i\leq \ell$ has at least one neighbour
		in $V_{k+1}\cup \dots \cup V_K$.
	\end{enumerate}
\end{itemize}
\end{definition}

\begin{figure}[ht]
\hspace*{\fill}
\includegraphics[scale=0.3,page=1]{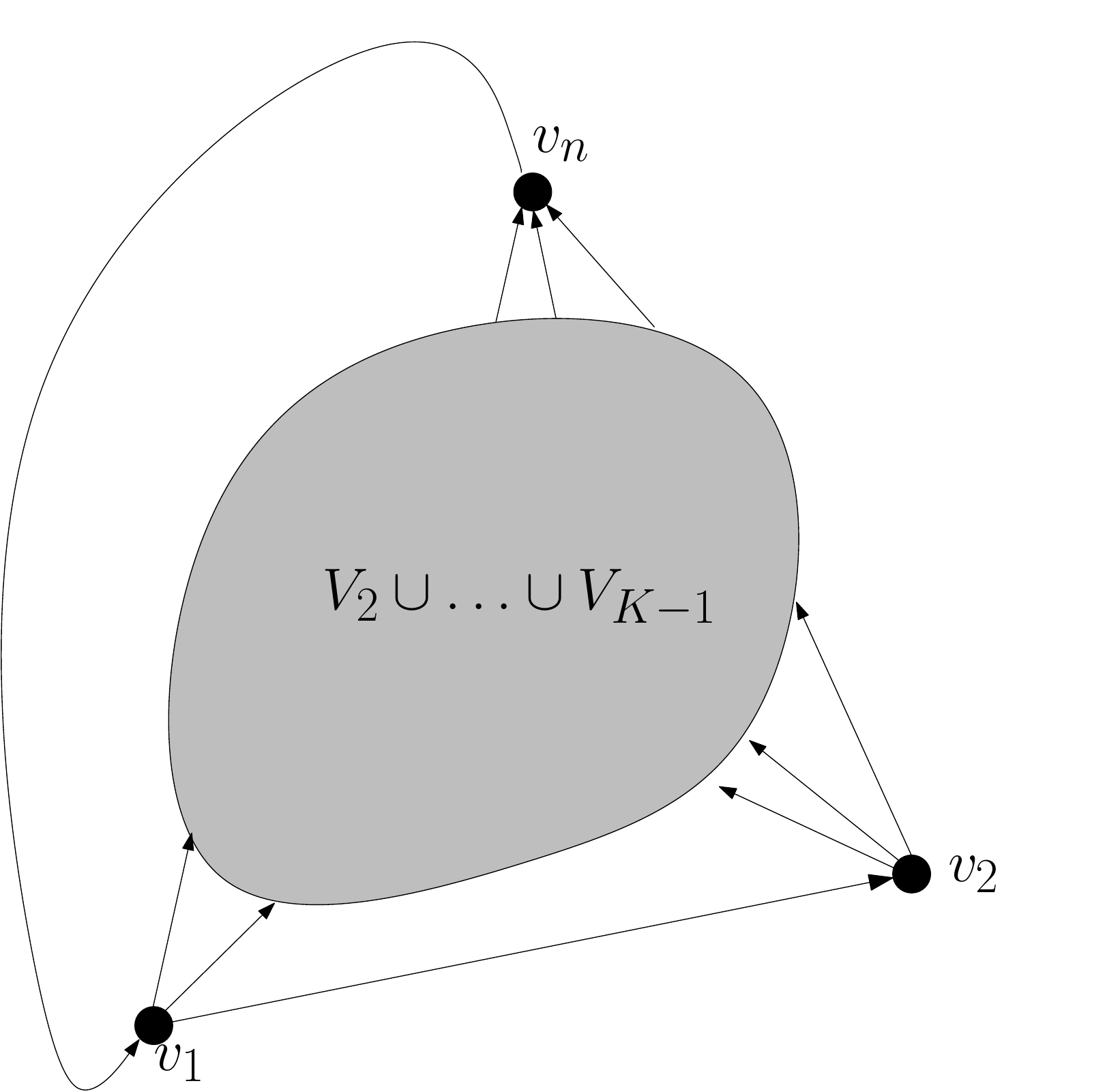}
\hspace*{\fill}
\includegraphics[scale=0.3,page=1]{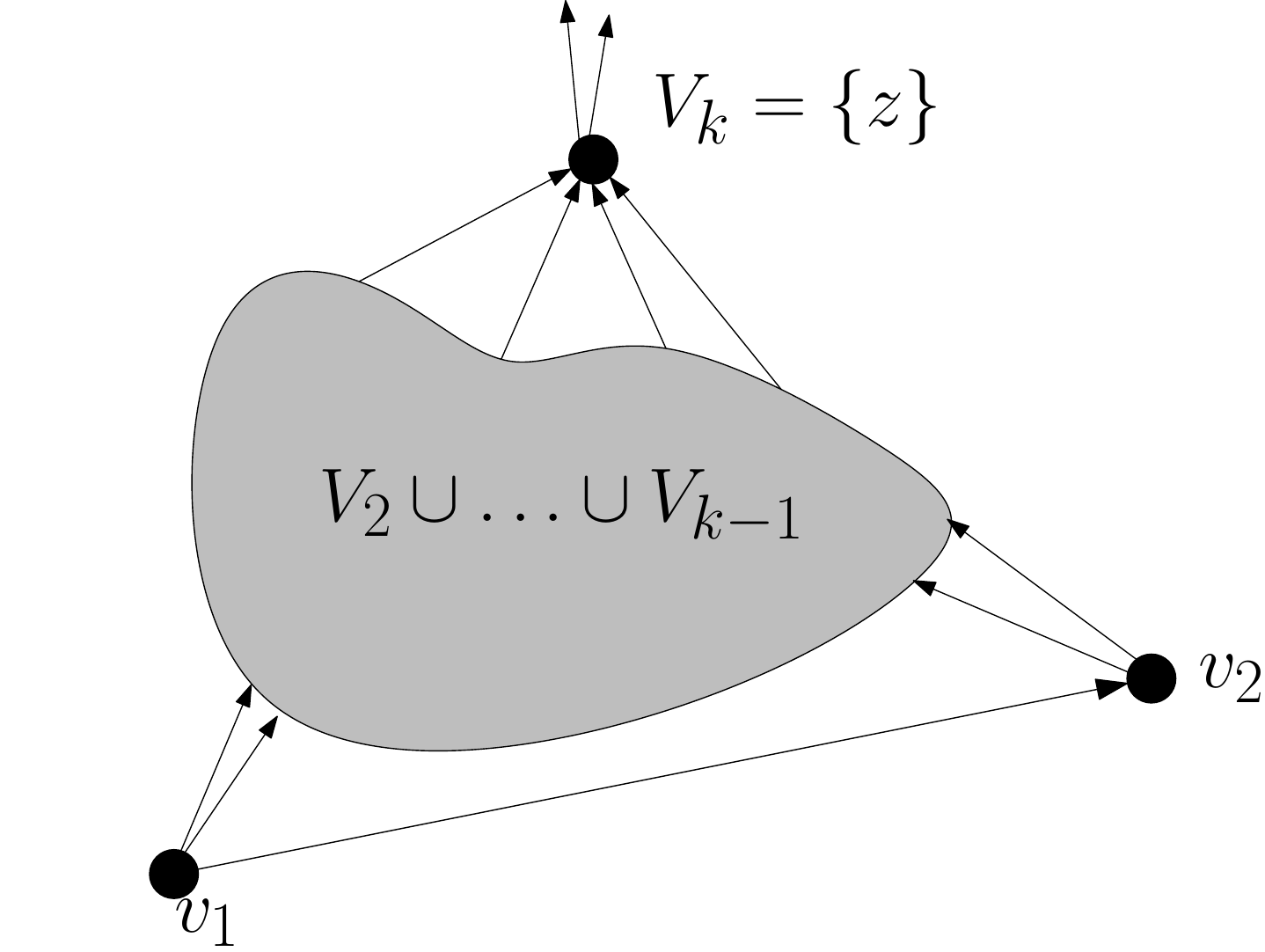}
\hspace*{\fill} \\
\hspace*{\fill} 
\includegraphics[scale=0.3,page=1]{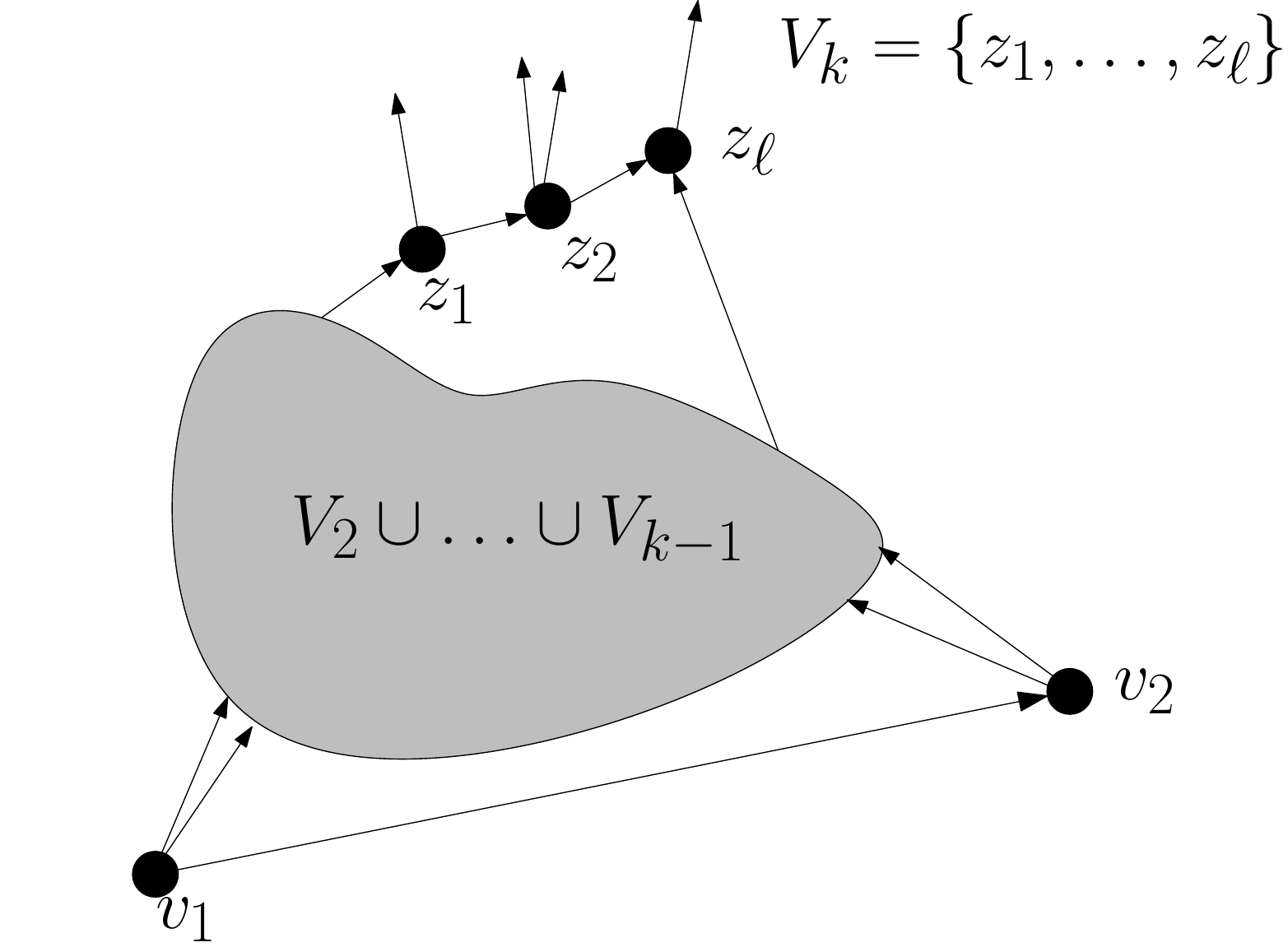}
\hspace*{\fill}
\includegraphics[scale=0.3,page=1]{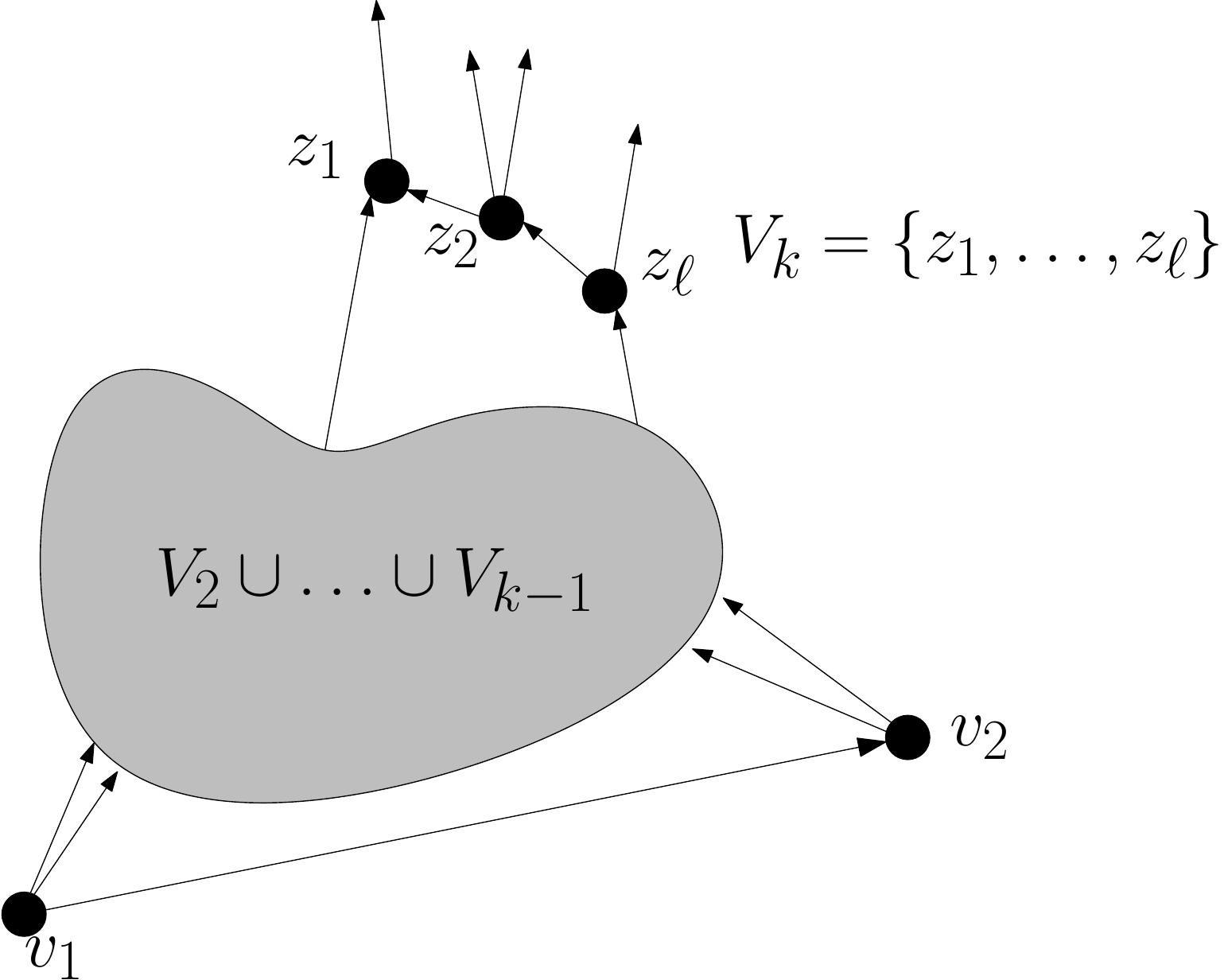}
\hspace*{\fill}
\caption{The canonical ordering with its implied edge directions
(defined in Section~\ref{se:edge_directions}.)
}
\label{fi:can_order}
\label{fig:can_order}
\end{figure}

Figure~\ref{fig:can_order} illustrates this definition.
A set $V_k$, $k=1,\dots,K$ is called a {\em group} of the canonical
ordering; a group with one vertex is a {\em singleton-group}, all other
groups are {\em chain-groups}.  
Edges with both ends in the same group are called {\em intra-edges}, all 
others are {\em inter-edges}.  Notice that when adding group $V_k$ for 
$k\geq 2$, there exists some faces (one for a chain-group, one or more for
a singleton-group) that are interior faces of $G[V_1\cup \dots \cup V_k]$
but were not interior faces of $G[V_1\cup \dots \cup V_{k-1}]$;
these faces are called the {\em faces completed by group $V_k$}.

Kant \cite{Kant96} showed that any 3-connected planar graph has such a canonical
ordering, even if the outer-face and the 2-path $v_n-v_1-v_2$ on
it to be used for the canonical ordering have been fixed.  
Furthermore, it can be found in linear time.

\subsection{Edge directions}
\label{se:edge_directions}

Given a canonical ordering, one naturally directs inter-edges from the
lower-indexed to the higher-indexed group.  For proving Barnette's theorem, it
will be useful to direct inter-edges as well as follows:

\begin{definition}
Given a canonical ordering, enumerate the vertices as
$v_1,\dots,v_n$ as follows.  
Group $V_1$ consists of $v_1$ and $v_2$.  For $2\leq k\leq K$,
let $s = |V_1|+\dots+|V_{k-1}|$.  
\begin{itemize}
\item If $V_k$ is a singleton group $\{z\}$,
	then set $v_{s+1}:= z$.
\item If $V_k$ is a chain-group
	$z_1,\dots,z_\ell$, then let $v_h$ and
	$v_i$ be the neighbours of $z_1$ and $z_\ell$ in 
	$V_1\cup\dots\cup V_{k-1}$,
	respectively.  
	If $h<i$, then set $v_{s+j}:=z_j$ for $j=1,\dots,\ell$,
	else set $v_{s+j}:=z_{\ell-j+1}$ for $j=1,\dots,\ell$.
\end{itemize}
\end{definition}

Let $\idx(v)$ be the index of vertex $v$ in this enumeration.  
Consider edges to be directed from
the lower-indexed to the higher-indexed vertex, with the
exception of edge $(v_1,v_n)$, which we direct $v_n\rightarrow v_1$.
These edge directions are illustrated in
Figure~\ref{fig:can_order}, with higher-indexed vertices drawn with
larger $y$-coordinate.

\begin{observation}
\label{obs:properties_directed}
(1) Every vertex has, in its clockwise order of incident edges,
a non-empty interval of incoming edges followed by a non-empty
 interval of outgoing edges.

(2) The edges on each of the two faces incident to $(v_1,v_n)$ form a directed
	cycle.

(3) For every face not incident to $(v_1,v_n)$, the incident edges form two 
	directed paths.
\end{observation}
\begin{proof}
For purposes of this proof only, consider edge $(v_1,v_n)$ to be directed
$v_1\rightarrow v_n$.  Then by properties of the canonical ordering,
every vertex except $v_1$ has at least one incoming edge, and
every vertex except $v_n$ has at least one outgoing inter-edge.
Therefore this orientation is {\em bi-polar}: it is acyclic 
with a single source $v_1$ and a single sink $v_n$.  
It is known \cite{TT86} that property (1) holds for all vertices $\neq v_1,v_n$
in a bi-polar orientation in a planar graph.  
Orienting edge $(v_1,v_n)$
as $v_n\rightarrow v_1$ also makes (1) hold at $v_1$ and $v_n$, since they
then have exactly one incoming/one outgoing edge.

In the bi-polar orientation, property (3) holds for any face $f$ \cite{TT86}.
Orienting edge $(v_1,v_n)$
as $v_n\rightarrow v_1$ will not change the property unless $f$ is incident
to $(v_1,v_n)$.  If $f$ is incident to $(v_1,v_n)$, then $v_1$ (as a source)
was necessarily the beginning and $v_n$ was necessarily the end of the two
directed paths.  
Orienting edge $(v_1,v_n)$
as $v_n\rightarrow v_1$ therefore turns the two directed paths into one
directed cycle.  So (2) holds.
\end{proof}

Define the {\em first} and {\em last}
outgoing edge to be the first and last edge in the clockwise order
around $v$ that is outgoing; this is well-defined by 
Observation~\ref{obs:properties_directed}(1).
Also define the following:

\begin{definition}
For any vertex $v_i$, $i\geq 2$, let the {\em parent-edge} be the
incoming edge $v_h\rightarrow v_i$ for which $h$ is maximized.
\end{definition}

If $e=v\rightarrow w$ is a directed edge, then $w$ is the {\em head}
of $e$, $v$ is the {\em tail} of $e$, and $v$ is a
{\em predecessor} of $w$.   The {\em left face} of $e$ is the face
to the left when walking from the tail to the head, and the {\em right
face} of $e$ is the other face incident to $e$.
The predecessor at the parent-edge of $w$ is called the {\em parent} of $w$.
The {\em predecessors of group $V_k$} are all vertices that
are predecessors of some vertex in $V_k$.

\subsection{Edge labels}
\label{se:labelling}

To read trees from the canonical ordering, it helps to assign labels 
to the edges incident to a vertex. 
They are very similar to Felsner's triorientation derived 
from Schnyder labellings \cite{Fel01} (which in turn can easily be derived from 
the canonical ordering \cite{MAN05}), but differ slightly in the handling 
of intra-edges and edge $(v_1,v_n)$.

\begin{definition}
Given a canonical ordering, label the edge-vertex-incidences as follows:
\begin{itemize}
\item Assume $V_k$ is a singleton-group $\{z\}$ with $2\leq k \leq K$.
\begin{itemize}
\item The first incoming edge of $z$ (in clockwise order) is labelled SE.
\item The last incoming edge of $z$ (in clockwise order) is labelled SW.
\item All other incoming edges of $z$ are labelled S.
\end{itemize}
\item Assume $V_k$ is a chain-group $\{z_1,\dots,z_\ell\}$ with $2\leq k < K$.
\begin{itemize}
\item The incoming inter-edge of $z_1$ is labelled SW at $z_1$.
\item The incoming inter-edge of $z_\ell$ is labelled SE at $z_\ell$.
\item Any intra-edge $(z_i,z_{i+1})$ is labelled E at $z_i$ and W at $z_{i+1}$.
\end{itemize}
\item Edge $v_1\rightarrow v_2$ is labelled E at $v_1$ and W at $v_2$.
\item Edge $v_n\rightarrow v_1$ is labelled S at $v_1$.  
\item If an inter-edge $v\rightarrow w$ is labelled SE / S / SW at $w$,
	then label it NW / N / NE at $v$.
\end{itemize}
Call an edge an {\cal L}-edge (for ${\cal L}\in \{S,SW,W,NW,N,NE,E,SE\}$) if it is 
labelled ${\cal L}$ at one endpoint.
\end{definition}
 
\begin{figure}[ht]
\hspace*{\fill}
\includegraphics[scale=0.4,page=2]{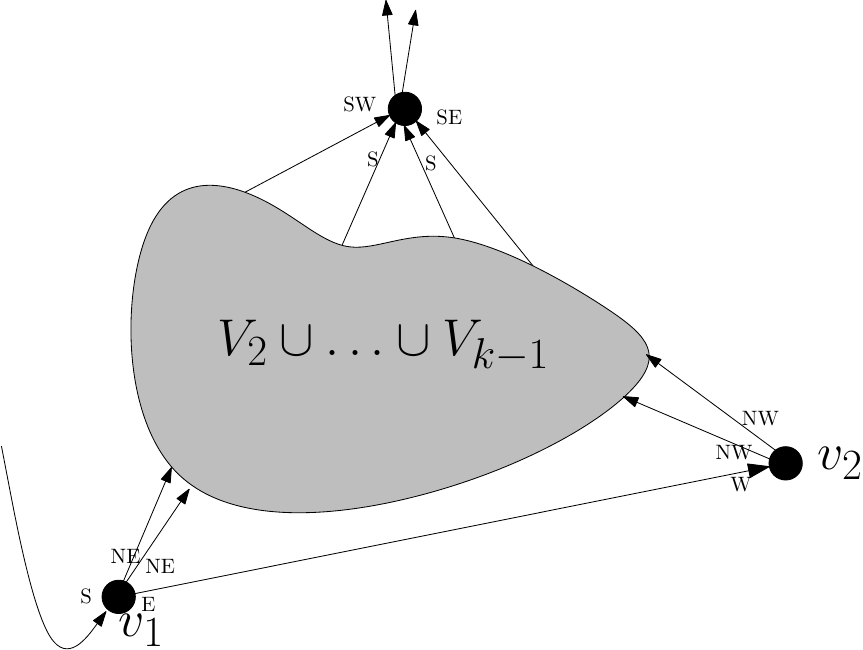}
\hspace*{\fill} 
\includegraphics[scale=0.4,page=2]{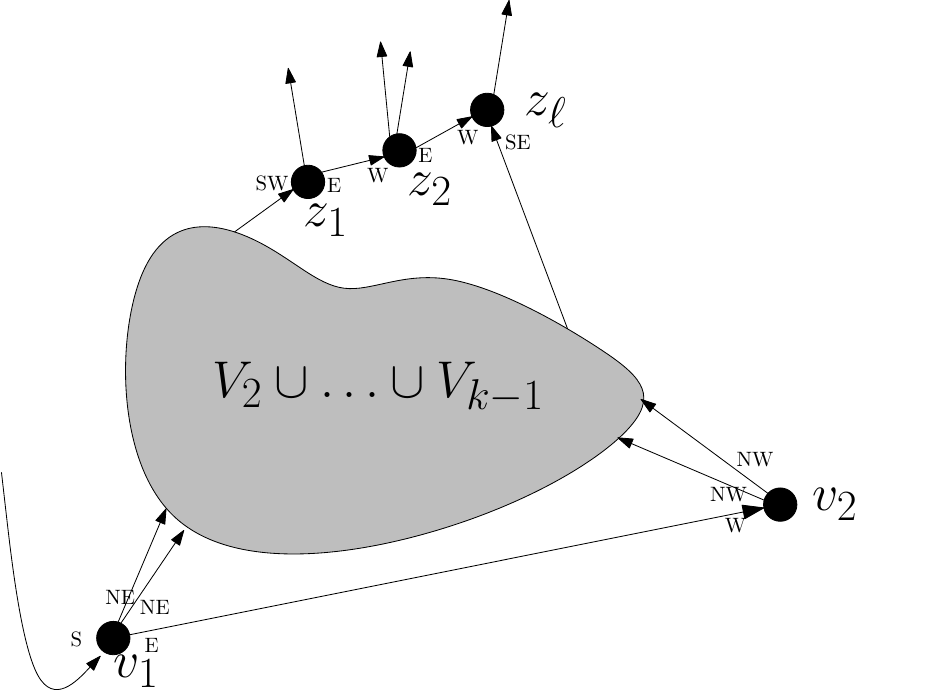}
\hspace*{\fill}
\caption{The canonical ordering with its implied edge labelling.}
\label{fig:can_order_label}
\end{figure}

See Figure~\ref{fig:can_order_label} for an illustration of this labelling.
The following properties are easily verified (see also \cite{CK97}
and \cite{Fel01} for similar results):

\begin{lemma}
\label{lem:label_properties}
\begin{itemize}
\item
At each vertex there are, in clockwise order, some edges labelled S, 
at most one edge labelled SW, at most one edge labelled W, 
some edges labelled NW, at most one
edge labelled N, some edges labelled NE, at most one edge labelled E,
and at most one edge labelled SE.
\item 
An edge is an intra-edge if and only if it is labelled E at one endpoint
and W at the other.
\item
No vertex has an edge labelled W and an edge labelled SW.
\item
No vertex has an edge labelled E and an edge labelled SE.
\end{itemize}
\end{lemma}

\section{Barnette's theorem via the canonical ordering}
\label{se:Barnette}

We now show that Barnette's theorem has a proof where
the tree can be read from a canonical ordering.

\begin{theorem}
\label{thm:Barnette}
Let $G$ be a planar graph with a canonical ordering.
Then the parent-edges forms a spanning tree of maximum degree 3. 
\end{theorem}
\begin{proof}
Let $T$ be the set of parent edges.
First note that each vertex $v_2,\dots,v_n$ has exactly one incoming
edge in $T$, and there is no directed cycle since $(v_1,v_n)$ is not
a parent-edge and therefore edges are directed
according to indices.  So $T$ is indeed a spanning tree.  
To see the bound on the maximum degree, the following lemma suffices:

\begin{lemma}
\label{lem:parent_first_last}
Assume $v\rightarrow w$ is a parent-edge of $w$.  Then either
$v\rightarrow w$ is the first outgoing edge at $v$ and labelled W or NW or N
at $v$,
or $v\rightarrow w$ is the last outgoing edge at $v$ and labelled E or NE or N
at $v$.
\end{lemma}
\begin{proof}
$w=v_1$ is impossible since $v_1$ has no parent.
If $w=v_2$, then its parent-edge $v_1\rightarrow v_2$ is the last
outgoing edge of $v_1$ and labelled W, so the claim holds.
Now consider $w=v_i$ for some $i\geq 3$, which means that
$w$ belongs to some group $V_k$ for $k\geq 2$.
There are two cases:
\begin{itemize}
\item $V_k$ is a chain-group $z_1-\dots-z_\ell$, which implies $k<K$.
	Assume that the chain is directed $z_1\rightarrow \dots
	\rightarrow z_\ell$; the other case is symmetric.  
	Note that $z_i$ is the parent of $z_{i+1}$
	for $1\leq i < \ell$, and $z_{i}\rightarrow z_{i+1}$ is the
	last outgoing edge of $z_i$ and labelled E, 
	so the claim holds for $w\in \{z_2,\dots,z_\ell\}$.

	Consider $w=z_1$.
	The parent $v$ of $z_1$ is the predecessor of $V_k$ adjacent to $z_1$.
	Let $x$ be the other predecessor of $V_k$ (it is adjacent to $z_\ell$).
	The direction of the chain implies $\idx(v)>\idx(x)$.
	Let $f$ be the face completed by $V_k$ and observe that it does
	not contain $(v_1,v_n)$.
	By Observation~\ref{obs:properties_directed}(3) the boundary of
	$f$ consists of two directed paths, which both end at $z_\ell$.  The 
	vertex where these two paths begin cannot be $v$, otherwise there 
	would be a directed path from $v$ to $x$ and therefore $idx(x)>idx(v)$.
	So $v$ has at least one incoming edge on face $f$, and hence
	$v\rightarrow z_1$ is its last outgoing edge.  Also,
	this edge is labelled SW at $z_1$, hence NE at $v$, as desired.

\item $V_k$ is a singleton-group $\{z\}$ with $z=w$.  Let $x\rightarrow w$ be
	an incoming edge of $w$ that comes before or after $v\rightarrow w$
	in the clockwise order of edges at $w$.  
	Such an edge must exist since $w$
	has at least two incoming edges (this holds for
	$w=v_n$ by 3-connectivity).  Assume that $x\rightarrow w$ comes
	clockwise before $v\rightarrow w$; the other case is similar.
	
	Let $f$ be the face incident to
	edges $v\rightarrow w$ and $x\rightarrow w$.
	By construction $f$ is not incident to $(v_1,v_n)$, and
	by Observation~\ref{obs:properties_directed}(3) the boundary of
	$f$ consists of two directed paths, which both end at $w$.  The 
	vertex where these two paths begin cannot be $v$, otherwise there 
	would be a directed path from $v$ to $x$, hence $\idx(x)>\idx(v)$
	contradicting the definition of parent-edge $v\rightarrow w$.
	So $v$ has at least one incoming edge on face $f$.
	$v\rightarrow w$ is its last outgoing edge at $v$.
	Furthermore, $v\rightarrow w$ cannot be labelled SE at $w$
	(since $x\rightarrow w$ comes clockwise before it), so it
	is labelled SW or S at $w$, hence NE or N at $v$ as desired.
\end{itemize}
\end{proof}

So in $T$, every vertex is incident to at most three edges: the parent-edge,
the first outgoing edge, and the last outgoing edge.
This finishes the proof of Theorem~\ref{thm:Barnette}.  
\end{proof}

In a later paper \cite{Bar92}, Barnette strengthened his own theorem
to show that in addition one can pick one vertex and require that it
has degree 1 in the spanning tree.  Using the canonical ordering allows
us to strengthen this result even further:  All vertices on one face
have degree at most 2, and two of them can be required to
have degree 1.

\begin{corollary}
Let $G$ be a planar graph with vertices $u,w$ on a face $f$ such
that $G\cup (u,w)$ is 3-connected.  Then $G$ has a spanning
tree $T$ with maximum degree 3 such that $\deg_T(u)=1=\deg_T(w)$,
and all other vertex $x$ on face $f$ have $\deg_T(x)\leq 2$.
\end{corollary}
\begin{proof}
Let $G^+=G\cup (u,w)$ and find a canonical
ordering of $G^+$ with $u=v_1$ and $w=v_n$.  Let $T$ be the
spanning 3-tree of $G^+$ obtained from the parent-edges; this will
satisfy all properties.

Observe that $(v_1,v_n)$ is not a parent-edge, so $T$ is a spanning
tree of $G$ as well.  Let $f_\ell$ and $f_r$ be the left and right
face of $v_n\rightarrow v_1$.  Both faces are completed by $V_K=\{v_n\}$.
It follows that any edge on $f_\ell$ (except $v_n\rightarrow v_1$) is
a SW-edge, because only such edges may have a not-yet-completed face
on their left.  Therefore
for any vertex $x\neq v_n$ on $f_\ell$ the first outgoing edge is labelled
NE and by Lemma~\ref{lem:parent_first_last} it does not belong to $T$.
So $\deg_T(x)\leq 2$ for all $x\in f_\ell$.  Similarly one shows that
$\deg_T(x)\leq 2$ for all $x\in f_r$.  Finally, $\deg_T(v_n)=1$ since
$v_n$ has no outgoing parent-edges, and $\deg_T(v_1)=1$ since all vertices
other than $v_2$ have higher-indexed predecessors.
\end{proof}

\section{On Gr\"{u}nbaum's conjecture}
\label{se:Gruenbaum}

Figure~\ref{fig:3TreeNotGood} shows
an example of a graph where the 3-tree from Theorem~\ref{thm:Barnette}
yields a co-tree with unbounded degree.  So unfortunately 
the proof of Theorem~\ref{thm:Barnette} does not help to solve
Gr\"{u}nbaum's conjecture.
In this section, we show that every planar 3-connected graph $G$ has
a spanning tree $T$ such that both $T$ and its co-tree $T^*$ 
are 5-trees.  Tree $T$ will again be read from the canonical ordering, but 
with a different approach.  
Assume throughout this section that a canonical order
of $G$ has been fixed.

\begin{figure}[ht]
\hspace*{\fill}
\includegraphics[scale=0.4,page=1]{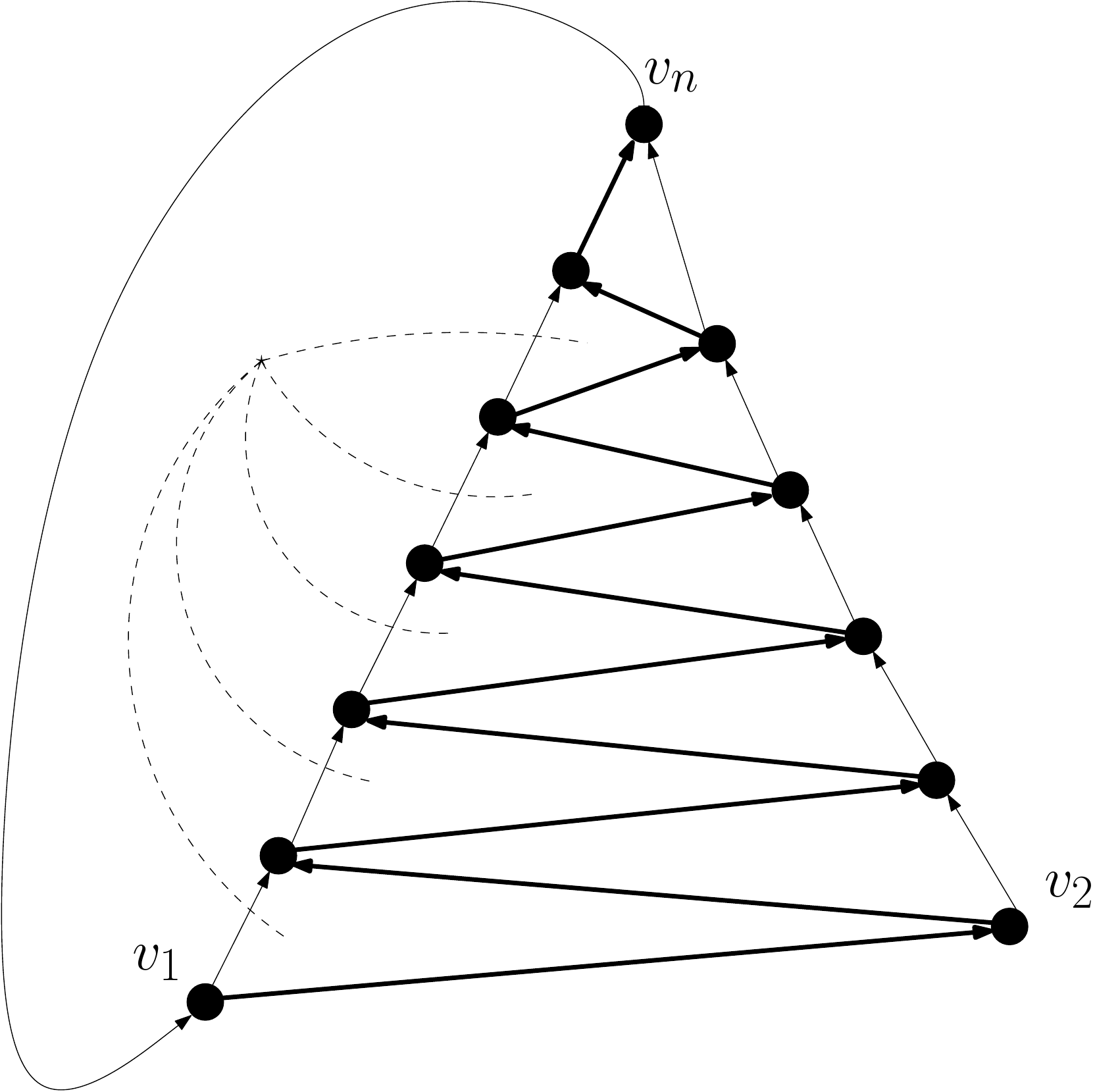}
\hspace*{\fill}
\caption{A planar 3-connected graph with the spanning 3-tree $T$ (bold) obtained
from the parent-edges.
Two faces have $\Theta(n)$ incident edges that are not in $T$. }
\label{fig:3TreeNotGood}
\end{figure}

\subsection{Dual canonical ordering}

A crucial insight is that a canonical ordering implies
a canonical ordering of the dual graph $G^*$.  This is
quite straightforward, but to our knowledge has not been published
explicitly.  An implicit proof follows from Felsner's result \cite{Fel04}
that the so-called Schnyder labelling of angles in a graph closely related
to $G$ corresponds to a Schnyder labelling of angles in a graph closely
related to $G^*$; combined with that Schnyder
labellings can be derived from canonical orderings and vice versa \cite{MAN05}.  
To avoid having to define these terms precisely, and because it will be 
important
how the edge labels of $G$ and $G^*$ relate, we give here an independent proof.

\begin{theorem}
\label{thm:dual_can_order}
\label{thm:dual_label}
For any canonical ordering of a 3-connected planar graph $G$,
there exists a canonical ordering of the dual graph $G^*$ such that
the following hold:
\begin{itemize}
\item The dual of any intra-edge of $G$ is a S-edge in $G^*$.
\item The dual of any S-edge of $G$ is an intra-edge in $G^*$.
\item The dual of any SW-edge $e$ of $G$ is a SE-edge in $G^*$,
	and directed from the left face of $e$ to the right face of $e$.
\item The dual of any SE-edge $e$ of $G$ is a SW-edge in $G^*$,
	and directed from the right face of $e$ to the left face of $e$.
\end{itemize}
\end{theorem}
\begin{proof}
Embed $G^*$ such that $(v_1,v_2)^*$ and $(v_1,v_n)^*$ are on the outer-face.
Define $f_1$ to be the outer-face of $G$, $f_2$ to be the interior face
incident to $(v_1,v_n)$, and set $F_1:=\{f_1^*,f_2^*\}$.  Observe that
$(f_1^*,f_2^*)=(v_1,v_n)^*$ is an edge on the outer-face of $G^*$ as desired.
Define $f_\phi$ to be the interior face incident to $(v_1,v_2)$, where
$\phi=m-n+2$ is the number of faces of $G$.  Define $F_k:=\{f_\phi^*\}$.
Observe that $(f_\phi^*,f_1^*)=(v_1,v_2)^*$ is an edge on the outer-face of
$G^*$ as desired.

For $3\leq k<K$, if $V_k$ is a chain-group, then
let $f$ be the face completed by $V_k$ and define 
$F_{K-k+2}:=\{f^*\}$.  If $V_k$ is a singleton-group, then
let $f_1,\dots,f_d$ be the faces completed by adding $V_k$,
enumerated in clockwise order around the unique vertex in $V_k$,
and define $F_{K-k+2}:=\{f_1^*,\dots,f_d^*\}$.  This finishes
the description of the canonical ordering of the dual graph.  
See also Figure~\ref{fig:dual_order}. 

It remains to verify the properties.   We will not show
2-connectivity of $G^*[F_1\cup\dots F_k]$ for $k\geq 2$ directly; this
holds because $F_1$ induces an edge and (as will be shown) for any $k\geq 2$ 
$F_k$ has at least two distinct
predecessors.

Consider the case where $F_k=\{f^*\}$ (for $2\leq k<K$) is a singleton-group.
This happens when $V_{k'}$, for $k'=K+2-k$, is a chain-group or a singleton-group
for which the vertex has two predecessors.  The latter can
be treated as a chain-group with only one vertex, so assume
that $V_{k'}=\{z_1,\dots,z_\ell\}$ (for $\ell\geq 1$) completed face $f$.
There are two kinds of edges at face $f$:  Those with both
endpoints in $V_1\cup \dots\cup V_{k'-1}$ (of which there is at least one
since $G[V_1\cup \dots \cup V_{k'}]$ is 2-connected), and
those with an endpoint in $V_{k'}$.  The latter edges all are on the
outer-face of $G[V_1\cup \dots \cup V_{k'}]$, hence their other incident
face is not yet completed and their dual edges are hence
incoming at $f^*$.  For the former edges the other face is 
completed (or, for $k'=2$, it is the outer-face), and so their dual edges 
are outgoing at $f^*$.   It follows that $f^*$ has at least one outgoing edge and
at least two incoming edges (since $V_{K+2-k}$ has at least two
predecessors.)  Moreover, the first incoming edge of $f^*$ is dual to
the SW-edge at $z_1$, and the last incoming edge of $f^*$ is dual to the
SE-edge at $z_\ell$.  This shows all claims for a singleton-group $F_k$.

Now consider the case where $F_k=\{f^*_{i_1},\dots,f_{i^*_\ell}\}$ for
some $\ell\geq 2$ is a chain-group.  This happens if $V_{k'}$, for $k'=K+2-k$,
is a singleton-group $\{z\}$, and $z$ had at
least three predecessors.  Since $F_k$ enumerates the faces in
clockwise order around $z$, they form a path in $G^*$.  For $1\leq h\leq \ell$,\
face $f_{i_h}$ has exactly two incident edges not in
$G[V_1\cup \dots \cup V_{k'-1}]$, namely, the two edges incident to $z$.
These are the only edges for which the duals could be incoming to
$f_{i_h}^*$.  But if $1<h<\ell$ then these dual edges are both 
intra-edges, so $f^*_{i_h}$ has no neighbour in $F_1\cup \dots \cup F_{k-1}$. 
For $h=1$ one of these dual edges is an intra-edge while the 
other is dual to the SE-edge of $z$.
Since the SE-edge of $z$ is on the outer-face of $G[V_1\cup \dots \cup V_{k'}]$,
its dual is incoming at $f^*_{i_1}$, hence $f^*_{i_1}$ has exactly one incoming
inter-edge as desired. By definition this edge is labelled SW at $f^*_{i_1}$.
Likewise $f^*_{i_\ell}$ has exactly one incoming inter-edge 
labelled SE at $f^*_{i_\ell}$.
Each $f_{i_h}$ (for $1\leq h\leq \ell$) has at least three edges, and only
two edges for which the duals are intra-edges or incoming inter-edges; 
so $f^*_{i_h}$ has at least one outgoing inter-edge as desired.

Finally notice that
the predecessors of $F_k$ must be two different vertices, otherwise
$z$ would either have no outgoing edge or would appear on a face twice
(which is impossible in a 3-connected graph.)  
\end{proof}

Call the canonical ordering of
$G^*$ obtained as in Theorem~\ref{thm:dual_can_order} the {\em dual
canonical ordering.}  

\begin{figure}[ht]
\hspace*{\fill}
\includegraphics[width=50mm,page=2]{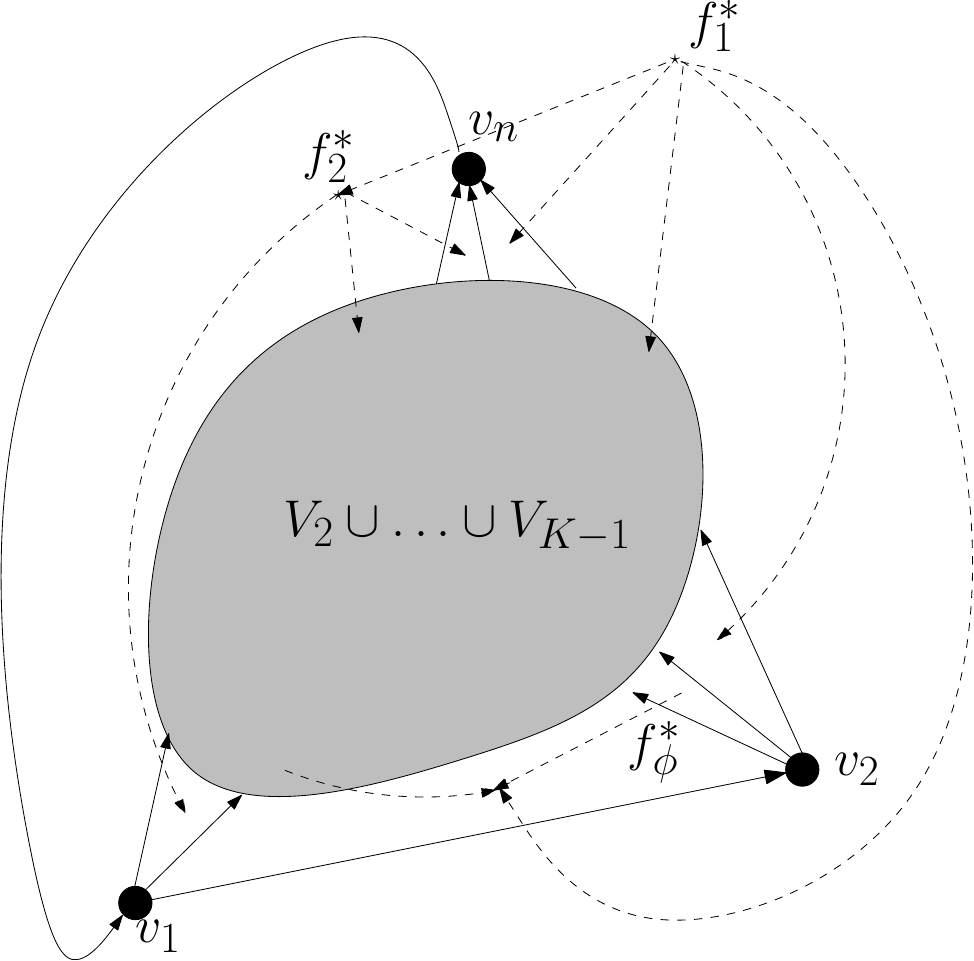}
\hspace*{\fill}
\includegraphics[width=50mm,page=3]{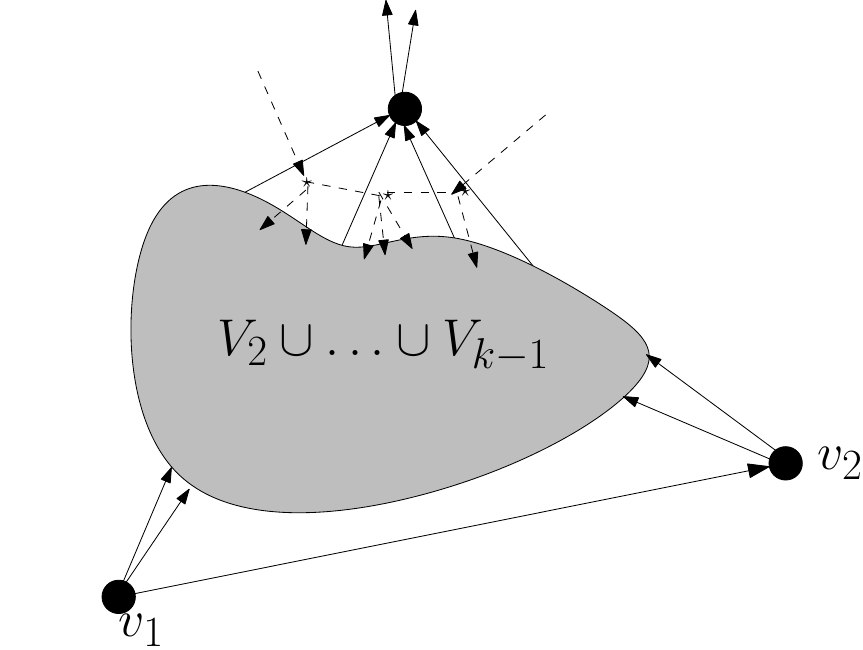}
\hspace*{\fill} 
\includegraphics[width=50mm,page=3]{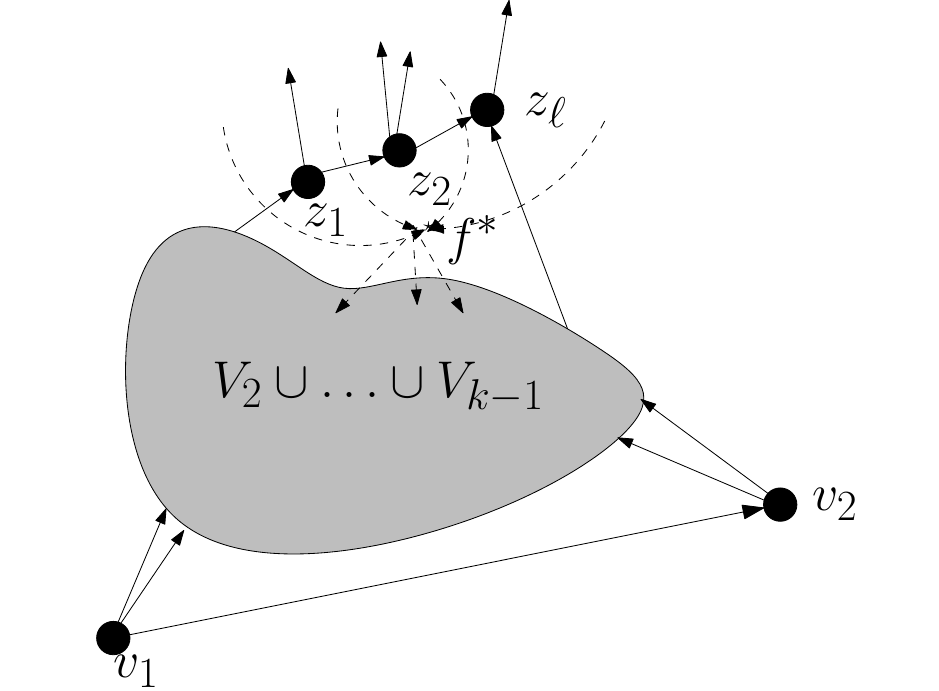}
\hspace*{\fill}
\caption{The dual canonical ordering obtained from a canonical ordering.}
\label{fi:dual_order}
\label{fig:dual_order}
\end{figure}

\subsection{The subgraph $H(G)$}

Now define a subgraph of $G$ from the labels of its edges.
If a vertex has NW-edges, then let the last one (in clockwise
order around $v$) be the {\em NNW-edge}.
Similarly define the  {\em NNE-edge} as the
first NE-edge in clockwise order.  

\begin{definition}
Presume a canonical ordering of a planar graph $G$ is fixed.  An edge $e$ of
$G$ is called an {\em $H$-edge} if it satisfies one of the following:
\begin{itemize}
\item[(H1)] $e$ is an intra-edge,
\item[(H2)] $e$ is the NNW-edge of its tail,
\item[(H3)] $e$ is the NNE-edge of its tail,
\item[(H4)] $e$ is the parent-edge of its head and the N-edge of its tail.
\end{itemize}
The graph formed by the $H$-edges of $G$ is denoted $H(G)$.
\end{definition}

\begin{figure}[ht]
\hspace*{\fill}
\includegraphics[scale=0.8,page=1]{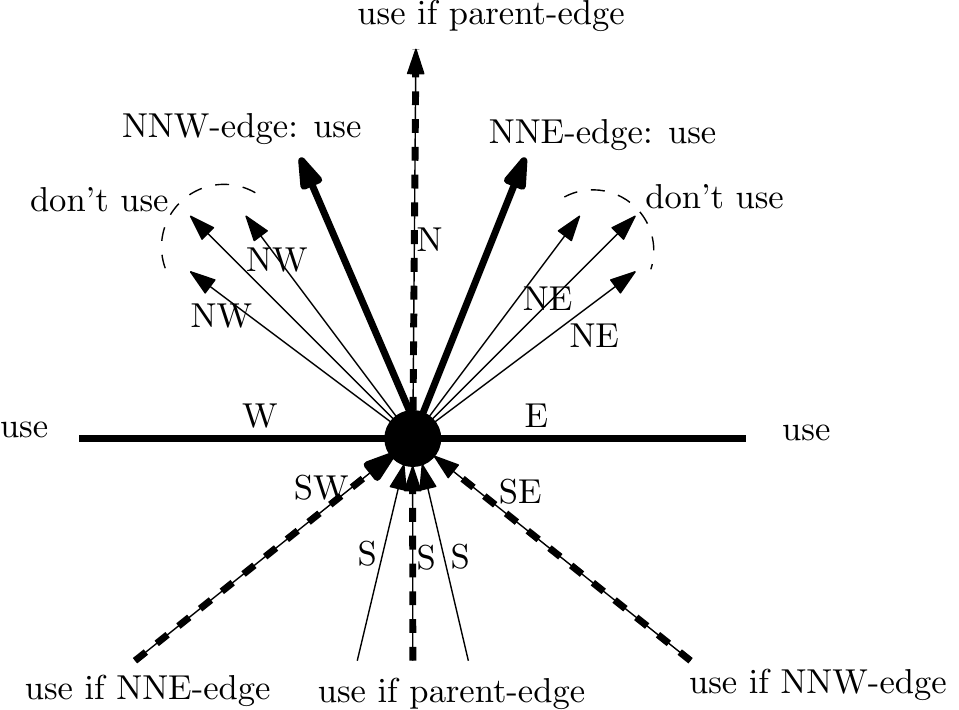}
\hspace*{\fill}
\caption{Illustration of H-edges.  Solid edges are H-edges; thick
dashed edges may be H-edges depending on the other endpoint.}
\label{fig:H-edges}
\end{figure}

\begin{lemma}
Any vertex $v$ has at most 5 incident $H$-edges.  
\end{lemma}
\begin{proof}
Observe first that $v$ has at most two incident $H$-edges that are outgoing 
inter-edges.  For no such edge is added under rule (H1).  Rules
(H2), (H3) and (H4) add at most one such $H$-edge each.  But if rule (H4)
adds edge $e$, then $e$ is the N-edge of $v$.  By Lemma~\ref{lem:parent_first_last}
it also is the first or last outgoing edge of $v$.  Therefore if rule (H4)
applies then $v$ has no NW-edge or no NE-edge, and so one of rules (H2) and (H3)
does not apply.  

Next consider the group of edges at $v$ consisting of the intra-edges at $v$,
and the SW-edge and SE-edge.  Clearly this group has at most four edges, but 
actually they are only two edges by Lemma~\ref{lem:label_properties}.
So $v$ has at most two incident $H$-edges in this group.

All edges at $v$ that are neither outgoing inter-edges nor in the above
group are incoming edges labelled $S$.  Such an edge is an $H$-edge
only if it is the parent-edge of $v$, so there is at most one $H$-edge among them.
So $v$ has at most 5 incident $H$-edges.
\end{proof}

Let $H(G^*)$ be the graph formed by the H-edges of $G^*$, using the dual
canonical ordering.   $H(G^*)$ also has maximum degree 5.  
Note that neither $H(G)$ nor $H(G^*)$ is necessarily a tree.  It is not even
obvious that they are connected (though in the following this will be shown to
hold.)
The plan is now to find a spanning tree of $H(G)$ for which the co-tree
belongs to $H(G^*)$.  Two lemmas are needed for this.

\begin{lemma}
\label{lem:dual_in_Hstar}
\label{lem:dual_Hstar}
Let $e$ be an edge in $G-H(G)$.  
Then the dual edge $e^*$ of $e$ belongs to $H(G^*)$.
\end{lemma}
\begin{proof}
If $e$ is a N-edge, then its dual is an intra-edge and hence belongs
to $H(G^*)$.  Edge $e$ cannot be a NNW-edge or NNE-edge or intra-edge 
since it is not in $H(G)$.
The remaining case is hence that $e$ is a NW-edge of its tail $v$, but not
the NNW-edge.  (The case of a NE-edge that is not the NNE-edge is similar.)
Figure~\ref{fig:proofs} (left) illustrates this case.

Let $e'$ be the clockwise next edge at $v$; this is also a NW-edge of $v$
since $e$ is not the NNW-edge.  Let $f$ be the face between $e$ and $e'$ at $v$.
By Theorem~\ref{thm:dual_label}, edge 
$(e')^*$ is labelled SW at $f^*$ while $e^*$
is labelled NE.  Since $e^*$ and $e'^*$ are consecutive at $f^*$,
therefore $e^*$ is the NNE-edge of $f^*$ and hence in $H(G^*)$.
\end{proof}

\begin{figure}[ht]
\hspace*{\fill}
\includegraphics[scale=0.4,page=2]{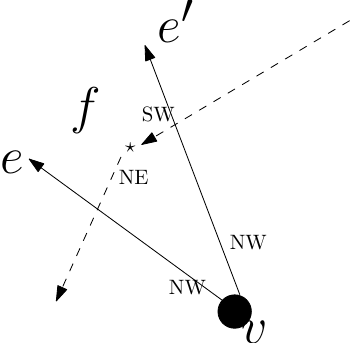}
\hspace*{\fill}
\includegraphics[scale=0.4,page=2]{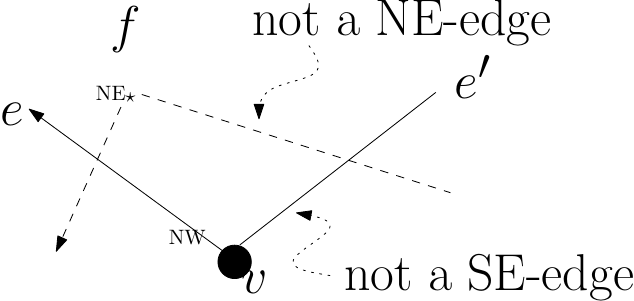}
\hspace*{\fill} 
\includegraphics[scale=0.4,page=2]{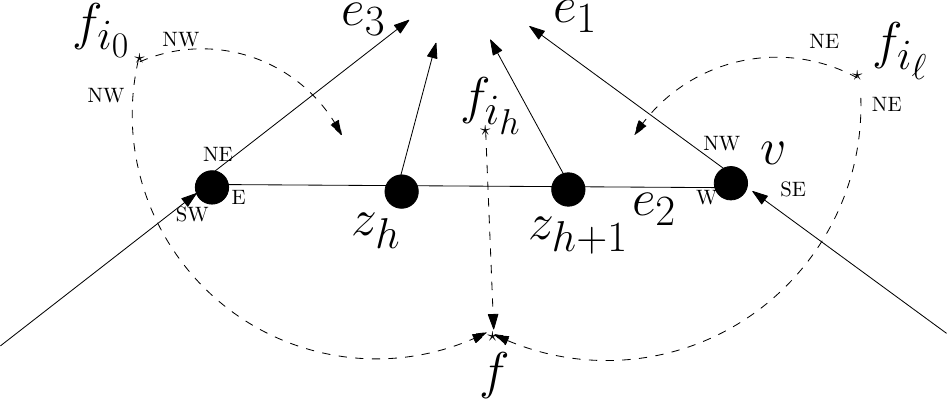}
\hspace*{\fill}
\caption{For the proofs of Lemma~\ref{lem:dual_in_Hstar} and
\ref{le:cycle_breaking}.}
\label{fig:proofs}
\end{figure}

\begin{lemma}
\label{le:cycle_breaking}
Let $C$ be a cycle of edges in $H(G)$.  Then there exists an edge $e\in C$
such that $e^*$ belongs to $H(G^*)$.
\end{lemma}
\begin{proof}
There are three cases where $e$ can be found easily; the bulk of the
proof deals with
the more complicated situation where none of them applies.
\begin{itemize}
\item[(C1)]  $C$ contains a N-edge $e$.  Then $e^*$ is an
intra-edge and belongs to $H(G^*)$ by rule (H1).
\item[(C2)] $C$ contains a NW-edge $e$ such that the clockwise
next edge $e'$ at $e$'s tail $v$ is not a SE-edge.  
This case is illustrated in Figure~\ref{fig:proofs}(middle).
Let $f$ be the face between $e$ and $e'$.
Since $e$ is a NW-edge, $e^*$ is a NE-edge.  Since
$e'$ is not a SE-edge, $(e')^*$ is not a NE-edge.  So
$e^*$ is the NNE-edge of $f^*$ and belongs to $H(G^*)$ by rule (H2).
\item[(C3)] $C$ contains a NE-edge $e$ such that the counter-clockwise
next edge at $e$'s tail is not a SW-edge.    With a symmetric argument
to (C2) one then shows that $e^*$ is a NNW-edge and belongs to $H(G^*)$ by 
rule (H3).
\item[(C4)] None of the above cases applies.  Since intra-edges form
paths, cycle $C$ must contain some inter-edges.  Let $e_1$ be the inter-edge
of $C$ that minimizes the index of its tail $v$.  $e_1$ is not a N-edge,
otherwise (C1) would apply.  So $e_1$ is either the NNW-edge or the NNE-edge
of $v$.  Assume that $e_1$ is a NNW-edge; the
other case is symmetric.  The goal is now to show that the situation is as
illustrated in Figure~\ref{fig:proofs}(right).

Let $e_2$ be the other edge in $C$ incident to $v$.  Edge $e_2$ cannot
be a N-edge at $v$, otherwise (C1) would apply.  It also cannot be a NE-edge
or E-edge at $v$, otherwise the clockwise edge after $e_1$ at $v$
is not a SE-edge and (C2) would apply.  Edge $e_2$ also cannot
be a SE-edge or S-edge or SW-edge at $v$, otherwise it would be an incoming inter-edge 
and its tail would have a smaller
index than $v$, contradicting the choice of $e_1$.  Also $e_2$ cannot
be a NW-edge at $v$, because the NNW-edge $e_1$ is the only NW-edge that is an
$H$-edge at $v$.
Thus edge $e_2$ must be an intra-edge labelled W at $v$.

Let $V_k=\{z_1,\dots,z_\ell\}$ be the chain-group containing edge $e_2$.
Notice that $v$ has no E-edge (otherwise (C2) would apply), 
so $v=z_\ell$.  Let $a$ be the minimal index such that
that path $z_a-z_{a+1}-\dots-z_\ell$ is part of $C$.  Let $e_3$ be the edge incident
to $z_a$ that is on $C$ and different from $(z_a,z_{a+1})$.  Observe that
$e_3$ is an inter-edge, for if it were an intra-edge then its other endpoint
would be $z_{a-1}$, contradicting the definition of $a$.  Also observe that
$e_3$ cannot be incoming at $z_a$, for otherwise the index of its tail
would be smaller than all indices in $V_k$, and in particular smaller than
the index of $v=z_\ell$; this contradicts the choice of $e_1$.

So $e_3$ is an outgoing inter-edge at $z_a$.  If $e_3$ were a N-edge then
(C1) would apply.  If it were a NW-edge, then (due to E-edge 
$(z_a,z_{a+1})$) (C2) would apply.  So $e_3$ is a NE-edge. Since it is
an $H$-edge, therefore is must be the NNE-edge of $z_a$.  Since (C3)
does not apply, $z_a$ cannot have a W-edge, which shows that $a=1$.

Let $f$ be the face completed by $V_k$, and let $f^*_{i_0},\dots,f^*_{i_\ell}$ 
be the predecessors of $f^*$ in the dual canonical order.  By the correspondence
of edge-label of Theorem~\ref{thm:dual_can_order}, 
$f_{i_0}$ shares the SW-edge of $z_1$ with $f$, face $f_{i_h}$ (for $1\leq h<\ell$) 
shares $(z_i,z_{i+1})$ with $f$, and $f_{i_\ell}$ shares the
SE-edge of $z_\ell$ with $f$.

Let $f^*_{i_p}\rightarrow f^*$ be the parent-edge of $f^*$ in the dual
canonical ordering.  Observe that $p\neq 0$.  For edge $(f^*_{i_0},f^*)$
is a NW-edge at $f^*_{i_0}$, as is $e_3^*$.  Thus  $(f^*_{i_0},f^*)$
is not the first outgoing edge at $f^*_{i_0}$, and by 
Lemma~\ref{lem:parent_first_last} hence not a parent-edge.
Likewise one shows $p\neq \ell$.  So $1\leq p< \ell$ and the parent-edge of $f^*$
is a N-edge.  By rule (H4) the parent-edge of $f^*$ is in 
$H(G^*)$.  Setting $e=(z_p,z_{p+1})$ hence yields the result.
\end{itemize}
\end{proof}

\subsection{Putting it all together}

\begin{theorem}
Every planar 3-connected graph $G$ has a spanning tree $T$ such that
both $T$ and its co-tree have maximum degree at most 5.  $T$
can be found in linear time.
\end{theorem}
\begin{proof}
First observe that $H(G)$ is connected.  For if it were disconnected,
then there would exist a non-trivial cut with all cut-edges in $G-H(G)$.
By Lemma~\ref{lem:dual_Hstar} the duals of the cut-edges belong to $H(G^*)$.  
Since cuts in a planar graph correspond to unions of cycles in the dual, 
hence the duals of the cut-edges contain a non-empty cycle $C$ of edges in $H(G^*)$. 
By Lemma~\ref{le:cycle_breaking} one edge of $C$ has
its dual in $H(G)$, contradicting the definition of the cut.

Let $H_0$ be all those edges in $H(G)$ for which the dual edge does not
belong to $H(G^*)$.    By Lemma~\ref{le:cycle_breaking} $H_0$ contains
no cycle, so it is a forest.  Assign a weight of 0 to all edges in $H_0$,
a weight of 1 to all edges in $H(G)-H_0$, and a weight of $\infty$ to
all edges in $G-H(G)$.  Then compute a minimum spanning tree $T$ of $G$.  Since
$H_0$ is a forest, all its edges are in $T$.  Since $H(G)$ is connected,
no edge in $G-H(G)$ belongs to $T$.  So $T$ is a subgraph of $H(G)$ and
has maximum degree at most 5.  All edges in the co-tree $T^*$ of $T$ are duals
of edges that are in $G-H_0$, and by definition of $H_0$ and
Lemma~\ref{lem:dual_Hstar} these edges belong to $H(G^*)$.  
So $T^*$ is a subgraph of
$H(G^*)$ and has maximum degree at most 5.  

It remains to analyze the time complexity.  One can compute a canonical
ordering in linear time, and from it, obtain the dual canonical ordering
and the edge-sets $H(G)$ and $H(G^*)$ in linear time.  The bottleneck
is hence the computation of the minimum spanning tree.  But there are only
3 different weights, and using a bucket-structure, rather than a priority
queue, in Prim's algorithm, we can find
the next vertex to add to the tree in constant time. Hence the minimum
spanning tree can be found in linear time.
\end{proof}

\section{Conclusion}

In this paper, we showed that every planar 3-connected graph has
a spanning tree of maximum degree 5 such that the co-tree also
has a spanning tree of maximum degree 5.  This is a first step
towards proving Gr\"unbaum's conjecture.

Barnette's theorem has as easy consequence that every planar
3-connected graph has a {\em 3-walk}: a walk that visits every
vertex at most 3 times.  But in fact, one can show a stronger
statement: Every planar 3-connected graph has a 2-walk \cite{GR94}.
The results in the paper imply that every planar 3-connected
graph has a 5-walk that also visits each face $f$ (in the sense of
``walking along part of the boundary of $f$'') at most 5 times.
A first step towards Gr\"unbaum's conjecture would be to try
to reduce  this ``5'' to a smaller number.
%

\bibliographystyle{plain}
\bibliography{../bib/journal,../bib/full,../bib/gd,../bib/papers}



\end{document}